\documentclass[11pt,reqno]{amsart}
\usepackage[T1]{fontenc}
\usepackage{lmodern}
\usepackage{amsmath,amssymb,mathtools,mathrsfs}
\usepackage{tikz}
\usepackage{float}
\usepackage{needspace}
\usepackage[margin=1in]{geometry}
\usepackage[expansion=false]{microtype}
\usepackage[colorlinks=true,linkcolor=blue,citecolor=blue,urlcolor=blue]{hyperref}
\allowdisplaybreaks[1]
\newtheorem{theorem}{Theorem}[section]
\newtheorem{proposition}[theorem]{Proposition}
\newtheorem{lemma}[theorem]{Lemma}
\newtheorem{corollary}[theorem]{Corollary}
\theoremstyle{remark}
\newcommand{\Hh}{\mathbb H}
\newcommand{\R}{\mathbb R}
\newcommand{\E}{\mathbb E}
\newcommand{\PP}{\mathbb P}
\newcommand{\Tt}{\mathbb T}
\newcommand{\ac}{\mathrm{ac}}
\newcommand{\dd}{\,\mathrm d}
\newcommand{\one}{\mathbf1}
\newcommand{\norm}[1]{\left\lVert#1\right\rVert}
\DeclareMathOperator{\Spec}{Spec}
\DeclareMathOperator{\Imn}{Im}
\DeclareMathOperator{\Bin}{Bin}
\title[Quantum percolation conjecture on regular trees]{Quantum percolation conjecture on regular trees}
\date{}
\hypersetup{pdfauthor={Simon Becker and Izak Oltman},pdftitle={Quantum percolation conjecture on regular trees},pdfsubject={Bernoulli bond percolation and spectral thresholds on regular trees}}
\begin{document}
\begin{abstract}
For Bernoulli bond percolation on every regular tree of degree at least
three, we prove that infinite clusters appear strictly before absolutely
continuous spectrum. Throughout an explicit supercritical interval, the
adjacency operator of every open cluster almost surely has no absolutely
continuous component. Together with Bordenave's theorem, this gives $p_c<p_\ac<1$. 
\end{abstract}
\author{Simon Becker}
\address{Bocconi University, Via Roentgen 1, 20136 Milan, Italy}
\email{simon.becker@unibocconi.it}
\author{Izak Oltman}
\address{Department of Mathematics, Northwestern University, 2033 Sheridan Road, Evanston, IL 60208, USA}
\email{ioltman@northwestern.edu}
\maketitle

\section{Introduction and main results}\label{sec:introduction}

Bernoulli bond percolation produces a random subgraph by retaining
each edge independently with probability $p$.
For $p$ larger than the percolation threshold $p_c$,
infinite connected components appear.
The quantum percolation problem asks whether these infinite components
also support extended quantum states, formulated here through the
presence of absolutely continuous spectrum
\cite{Anderson,deGennesLaforeMillot,Bordenave}.
On every regular tree of degree at least three, we prove that the two
phenomena occur at different thresholds: there is an explicit interval
above the percolation threshold in which every open cluster almost
surely has no absolutely continuous spectral component.

Fix an integer $k\geq2$. Let $G=(V,\mathcal E)=\Tt_{k+1}$ be the infinite
$(k+1)$-regular tree, with vertex set $V$ and edge set $\mathcal E$.
For $p\in[0,1]$, retain each edge independently with probability $p$.
Write $\omega=(\omega_e)_{e\in\mathcal E}$, where $\omega_e$ is one for
an open edge and zero otherwise, and denote the product probability and
expectation by $\PP_p$ and $\E_p$. For neighboring vertices $x,y$, write
$y\sim x$ and $\omega_{xy}=\omega_{\{x,y\}}$.
The adjacency operator on the space $\ell^2(V)$ of square-summable
complex functions is
\begin{equation}\label{eq:model}
 (A_\omega u)(x)=\sum_{y\sim x}\omega_{xy}u(y),
 \qquad u\in\ell^2(V),\quad x\in V.
\end{equation}
This operator is bounded and self-adjoint, with operator norm
$\|A_\omega\|\leq k+1$. Write $\mathcal H_\ac(A_\omega)$ for its
absolutely continuous subspace: the vectors whose spectral measures are
absolutely continuous with respect to Lebesgue measure.

\begin{theorem}\label{thm:tree}
For every $k\geq2$ there exists an explicit $\varepsilon_k>0$ such that
\begin{equation*}
 \frac1k<p\leq\frac1k+\varepsilon_k
 \quad\Longrightarrow\quad
 \mathcal H_\ac(A_\omega)=\{0\}\qquad\PP_p\text{-almost surely}
\end{equation*}
on $\Tt_{k+1}$. The constant $\varepsilon_k$ is given in \eqref{eq:epsilon}.
\end{theorem}

Fix a root $o\in V$. Its open cluster $C_o$ consists of the vertices
connected to $o$ by open paths and the edges of those paths. Write
$|C_o|$ for its number of vertices and $A_{C_o}$ for its adjacency
operator. Define the connectivity threshold and the onset of absolutely
continuous spectrum by
\begin{align}
 p_c(G)&=\inf\{p\in[0,1]:\PP_p(|C_o|=\infty)>0\},\label{eq:pc}\\
 p_\ac(G)&=\inf\left\{p\in(p_c(G),1]:
 \PP_p\!\left(\mathcal H_\ac(A_{C_o})\ne\{0\}
        \,\middle|\,|C_o|=\infty\right)>0\right\},
 \qquad\inf\varnothing=+\infty.\label{eq:onset}
\end{align}
We omit $G$ when the graph is fixed. For an integer $n\geq0$ and
$a\in[0,1]$, let $\Bin(n,a)$ denote the binomial distribution with $n$
trials and success probability $a$. The forward open descendants have
offspring distribution $\Bin(k,p)$, with mean $kp$. The Galton--Watson
extinction criterion therefore gives $p_c=1/k$
\cite[Proposition~5.4 and Section~5.2]{LyonsPeres}.
Bordenave proved the existence of an absolutely continuous component
for $p$ sufficiently close to one
\cite[Theorem~3 and Corollary~4]{Bordenave}. Thus our theorem gives the
remaining strict inequality in the following statement.

\begin{corollary}\label{cor:threshold}
With the constant in Theorem~\ref{thm:tree},
\begin{equation}\label{eq:separation}
 p_c(\Tt_{k+1})=\frac1k
 <\frac1k+\varepsilon_k\leq p_\ac(\Tt_{k+1})<1,
 \qquad k\geq2.
\end{equation}
\end{corollary}

Our contribution is the lower separation $p_\ac>p_c$; the upper bound
$p_\ac<1$ is Bordenave's. The threshold in \eqref{eq:onset} is an onset,
not a claim that spectral type is monotone in $p$; compare the two
threshold conventions in \cite[Section~1.2]{Bordenave}.
The quantum percolation statement proved here is spectral: absence of
an absolutely continuous component is not a claim of pure-point spectrum
or dynamical localization.
The constants below are explicit but not optimized; their role is to prove
strict separation of the two thresholds, not to locate the transition sharply.

For each fixed $p$ in Theorem~\ref{thm:tree}, one probability-one set works
for all open clusters. Indeed, if $\mathcal C(\omega)$ is their collection
and $V(C)$ is the vertex set of $C$, then
\[
 \ell^2(V)=\bigoplus_{C\in\mathcal C(\omega)}\ell^2(V(C)),\qquad
 A_\omega=\bigoplus_{C\in\mathcal C(\omega)}A_C.
\]
For $p>p_c$, the root cluster is infinite with positive probability.
Choose $p_{\rm B}(k)<1$ so that Bordenave's conclusion holds whenever
$p_{\rm B}(k)<p\leq1$. Figure~\ref{fig:threshold-line} locates the two
proved regimes without assigning a spectral type to the region between
them.

\begin{figure}[htbp]
\centering
\begin{tikzpicture}[x=1cm,y=1cm]
  \draw[->,line width=0.6pt] (0,0)--(10.45,0) node[right] {$p$};

  \foreach \x in {0,2.55,4.45,7.9,10}
    \draw[line width=0.6pt] (\x,-0.11)--(\x,0.11);

  \node[below=3pt] at (0,0) {$0$};
  \node[below=3pt,align=center] at (2.55,0) {$p_c=\dfrac1k$};
  \node[below=3pt,align=center] at (4.45,0) {$\dfrac1k+\varepsilon_k$};
  \node[below=3pt,align=center] at (7.9,0) {$p_{\rm B}(k)$};
  \node[below=3pt] at (10,0) {$1$};

  \draw[line width=2.2pt] (0,0.55)--(2.55,0.55);
  \node[above=2pt,align=center] at (1.275,0.55)
    {\scriptsize finite clusters a.s.\\[-1pt]\scriptsize $\Rightarrow$ no a.c.};

  \draw[line width=2.2pt] (2.55,0.55)--(4.45,0.55);
  \node[above=2pt,align=center] at (3.5,0.55)
    {\scriptsize Theorem~\ref{thm:tree}\\[-1pt]\scriptsize no a.c.};

  \draw[gray,densely dashed,line width=1.0pt] (4.45,0.55)--(7.9,0.55);
  \node[above=2pt] at (6.175,0.55) {\scriptsize not determined};

  \draw[line width=2.2pt] (7.9,0.55)--(10,0.55);
  \node[above=2pt,align=center] at (8.95,0.55)
    {\scriptsize Bordenave\\[-1pt]\scriptsize nontrivial a.c.};
\end{tikzpicture}
\caption{The two proved regimes, shown schematically and not to scale.
All clusters are finite for $p\leq p_c$. Theorem~\ref{thm:tree} rules out
an absolutely continuous component on every cluster through
$p_c+\varepsilon_k$. Bordenave's theorem gives a nontrivial absolutely
continuous component with positive probability for $p>p_{\rm B}(k)$.
The intervening region is not classified; no monotonicity of spectral
type is asserted.}
\label{fig:threshold-line}
\end{figure}
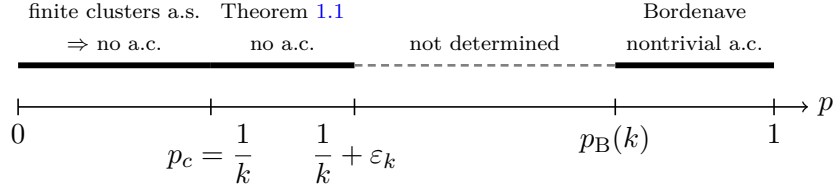

\subsection{The strategy of the proof}
The absolutely continuous part of the spectrum is detected by the
imaginary part of the boundary value of the diagonal resolvent, so it is
enough to prove that
\[
\operatorname{Im} G(E+i0)=0
\]
for almost every energy $E$. On a tree, the resolvent satisfies a
recursive relation, and after conditioning on survival this yields a
fixed-point equation for the law $\nu$ of the forward resolvent in the
upper half-plane,
\[
\nu=(1-q)P\nu+qQ\nu,
\]
where $P$ corresponds to exactly one infinite child branch and $Q$ to at
least two infinite child branches. The randomness of the finite side
branches makes $P^2$ a strict contraction on $L^2(\mathbb H)$: already
the two configurations consisting of no finite side branch and one
single-vertex side branch produce distinct M\"obius transformations whose
relative action yields a free-group spectral gap. The operator $Q$ need
not be contractive, but its contribution is multiplied by
\[
q=\PP(K\ge2\mid |T|=\infty),
\qquad
q\le kp-1,
\]
which is arbitrarily small when $p\downarrow 1/k$. After smoothing $\nu$
on the hyperbolic plane, the fixed-point equation therefore gives
\[
\|f\|_2
\le
(\rho+2qB_k)\|f\|_2
<
\|f\|_2,
\]
a contradiction; hence the resolvent boundary values are real almost
everywhere and the absolutely continuous spectrum is absent.

\subsection{Relation to earlier work}
The two sides of \eqref{eq:separation} concern different regimes.
Keller \cite{Keller} proved stability of absolutely continuous spectrum
for random trees under an assumption excluding leaves. Bordenave
\cite{Bordenave} allowed leaves and finite attached branches, including
the binomial offspring distributions arising from bond percolation.
Arras and Bordenave \cite{ArrasBordenave} gave quantitative criteria in
terms of offspring fluctuations. These results concern trees close to
a regular tree. Our theorem concerns the opposite regime: probabilities
just above the first appearance of infinite clusters.

Finite configurations have long been known to produce localized
eigenfunctions \cite{KirkpatrickEggarter,ChayesEtAl,Veselic}.
For bond percolation on the three-regular tree, Carmona, Klein and
Martinelli \cite[Section~7]{CarmonaKleinMartinelli} exhibited a
deterministic countable set, dense in the spectrum, of almost-sure
eigenvalues. Such eigenvalues can coexist with continuous spectrum.
The role of finite branches here is different: their randomness yields
an estimate excluding the entire absolutely continuous component, not
merely constructing individual eigenfunctions.

The Bernoulli variables have no density, so averaging arguments that
require a density cannot be used directly. Localization for singular
random potentials was proved in one dimension by Carmona, Klein and
Martinelli \cite{CarmonaKleinMartinelli}; higher-dimensional continuum
results include Bourgain and Kenig \cite{BourgainKenig} and Germinet and
Klein \cite{GerminetKlein}. For lattice Anderson--Bernoulli operators,
edge localization was established in two dimensions by Ding and Smart
\cite{DingSmart}, in three dimensions by Li and Zhang \cite{LiZhang},
and in all higher dimensions in the recent preprint of Li, Liu and
Zhang \cite{LiLiuZhang}. Li \cite{LiLargeDisorder} also treated large
disorder in two dimensions away from neighborhoods of finitely many
exceptional energies. Those works concern random potentials with fixed
hopping. Here the edges themselves are random, and the conclusion is
absence of absolutely continuous spectrum throughout the spectrum,
not pure-point or dynamical localization.

The resolvent recursion goes back to work of Abou-Chacra, Thouless and
Anderson \cite{AbouChacraThoulessAnderson} and, for quantum percolation,
Harris \cite{Harris}. On fixed trees, absolutely continuous spectrum and
its stability were studied in
\cite{Klein,AizenmanSimsWarzel,FroeseHaslerSpitzer,KellerLenzWarzel};
resonant delocalization is treated in
\cite{AizenmanWarzel,AizenmanWarzelBook}.
The harmonic analysis of fractional-linear averages appears in
Shubin, Vakilian and Wolff \cite{ShubinVakilianWolff} and Wolff
\cite{WolffContraction}. Here the explicit two-configuration calculation
and Kesten's estimate give the uniform contraction needed to compare
finite-branch randomness with rare infinite branching. The ball-average
convolution bound permits the comparison for singular resolvent laws.

There is also a distinction between spectral type in a realization and
regularity after averaging. Let $\mu_o$ be the root spectral measure
of $A_{C_o}$ at $\delta_o$, and set
$\overline\mu_{o,p}(B)=\E_p[\mu_o(B)]$ for Borel sets $B\subset\R$.
By \cite[Corollary~1.7]{BordenaveSenVirag}, this expected measure has a
nonzero continuous (non-atomic) part for every $p>1/k$.
Thus such a part is already present throughout our interval of
almost-sure absence of absolutely continuous spectrum. Continuity
here does not mean absolute continuity. Results on spectral tails,
including those for percolation Laplacians on regular trees
\cite{Reinhold,MullerStollmannSurvey}, address a different question:
the amount of spectral mass near an edge, rather than its spectral type.

Section~\ref{sec:analytic} proves the smoothing, group and convolution
estimates and extracts the fixed-point exclusion criterion.
Section~\ref{sec:trees} derives the conditioned resolvent law, computes
the contraction from the two finite configurations, and proves
Theorem~\ref{thm:tree} and Corollary~\ref{cor:threshold}.

\section{Notation and analytic estimates}\label{sec:analytic}
The proof compares a linear average of real fractional-linear maps
with a nonlinear operation adding independent upper-half-plane
variables. We first regularize arbitrary probability laws without
losing the group action, then prove the two estimates needed for the
comparison. The final proposition records how these estimates exclude
a stationary probability law in the upper half-plane.

\subsection{Notation}\label{sec:notation}
We write $\Hh=\{z\in\mathbb C:\Imn z>0\}$, $\one_B$ for an
indicator, and $\mathcal L(X)$ for the distribution of a random variable.
The notation $X\overset d=Y$ means equality of distributions, and
$\mathcal L(X\mid B)$ denotes conditioning on an event of positive
probability. A superscript $\otimes j$ denotes a $j$-fold product measure.
We use natural logarithms unless a base is displayed.

The percolation probability space is
$\Omega_G=\{0,1\}^{\mathcal E}$ with its product $\sigma$-algebra and
measure $\PP_p=\bigotimes_{e\in\mathcal E}\operatorname{Bernoulli}(p)$;
$\E_p$ denotes expectation. We omit $p$ when it is fixed.
For vertices $o,x$, let $[o,x]$ be the edge set of their unique simple
path. Then $o\leftrightarrow x$ means that all its edges are open, and
$\one_{\{o\leftrightarrow x\}}=\prod_{e\in[o,x]}\omega_e$, with empty
product one. In particular, the root cluster is
\begin{equation}\label{eq:cluster}
 \begin{split}
 V_o(\omega)&=\{x\in V:o\leftrightarrow x\} 
\text{ and }
 C_o(\omega)=\bigl(V_o(\omega),\{\{x,y\}\in\mathcal E:
       x,y\in V_o(\omega),\ \omega_{xy}=1\}\bigr).
 \end{split}
\end{equation}
For a graph $X$, its vertex and edge sets are $V(X)$ and
$\mathcal E(X)$, its adjacency operator is $A_X$, and
$|X|=|V(X)|$. We write $\deg_X(v)$ for the degree of $v$.
For a cluster $C$, $A_C$ is the restriction of $A_\omega$ to the
reducing subspace $\ell^2(V(C))$.

A superscript $*$ denotes the adjoint.
Inner products are conjugate-linear in the first variable;
$\delta_x$ is the coordinate vector at $x$, and $I$ is the identity.
When used as a measure, $\delta_a$ is the unit point mass at $a$.
For a self-adjoint operator $A$, write $\Spec(A)$ for its spectrum,
$\one_B(A)$ for its spectral projections, and
$\mu_u(B)=\langle u,\one_B(A)u\rangle$ for the spectral measure of $u$.
We denote its absolutely continuous part by $\mu_{u,\ac}$ and the
projection onto $\mathcal H_\ac(A)$ by $\Pi_\ac(A)$.
The notation $g(E+i0)$ means the finite limit
$\lim_{\tau\downarrow0}g(E+i\tau)$.

For a graph automorphism $\sigma$, set
$(\sigma\omega)_{\sigma x,\sigma y}=\omega_{xy}$ and
$(Q_\sigma u)(x)=u(\sigma^{-1}x)$. The product measure is invariant
under $\sigma$, and
$A_{\sigma\omega}=Q_\sigma A_\omega Q_\sigma^*$ and
$C_{\sigma o}(\sigma\omega)=\sigma C_o(\omega)$.
Transitivity of the regular tree therefore makes both thresholds
independent of the root.

\subsection{Boundary values and spectral measures}\label{sec:boundaryfacts}
For a bounded self-adjoint $A$ and a unit vector $e$, the scalar
resolvent is the Herglotz function
\[
 g_e(z)=\langle e,(A-z)^{-1}e\rangle
       =\int_\R\frac{d\mu_e(\lambda)}{\lambda-z},\qquad
 \Imn g_e(z)>0\quad(z\in\Hh).
\]
Its finite boundary values exist for Lebesgue-almost every energy by
\cite[Corollary~3.25]{Teschl}. We record how this fact applies to
countably many random operators.

On a probability space $(\Omega,\mathscr A,\PP)$, let $(A_j)_{j\in J}$
be a countable family of random self-adjoint operators on fixed spaces
$\ell^2(X_j)$, where each $X_j$ is countable. Assume measurable matrix
entries and a common deterministic bound $\|A_j\|\leq D$, with $D>0$.
Fix coordinate unit vectors $e_j$ and set
$g_j(z)=\langle e_j,(A_j-z)^{-1}e_j\rangle$.
Finite coordinate compressions converge strongly to $A_j$, hence
their resolvents converge strongly for $z\in\Hh$. It follows that
$(\omega,z)\mapsto g_j(z,\omega)$ is jointly measurable.
The Cauchy criterion along positive rational $\tau$ makes the set
$\mathcal B\subset\R\times\Omega$ where at least one finite boundary
value fails measurable; continuity for $\tau>0$ gives the same limit
along all real $\tau\downarrow0$.
By the cited boundary theorem and countability of $J$, each section
$\mathcal B(\omega)$ has Lebesgue measure zero. Thus Tonelli gives
\begin{equation}\label{eq:generalboundaryfubini}
 \int_{-D}^D\PP\bigl((E,\omega)\in\mathcal B\bigr)\,dE
 =\int |\mathcal B(\omega)\cap[-D,D]|\,d\PP(\omega)=0.
\end{equation}
Consequently, for any fixed countable set $\Sigma\subset\R$, there is
a deterministic Borel set $\mathcal E_{\rm bd}\subset[-D,D]\setminus\Sigma$
of full Lebesgue measure such that, for every $E\in\mathcal E_{\rm bd}$,
all $g_j(E+i0)$ are finite on one probability-one set.

We also use the spectral density formula
\cite[Theorem~3.23 and Section~3.3]{Teschl}:
\begin{equation}\label{eq:acdensitycriterion}
 \frac{d\mu_{e,\ac}}{dE}(E)=\frac1\pi\Imn g_e(E+i0)
 \quad\text{for Lebesgue-almost every }E.
\end{equation}
If these boundary values are real almost everywhere for every vector
in a countable set with dense linear span, then
$\|\Pi_\ac(A)e\|^2=\mu_{e,\ac}(\R)=0$ for each such vector, and
$\mathcal H_\ac(A)=\{0\}$.

\subsection{Averaging over hyperbolic balls}\label{sec:averaging}
We use the standard hyperbolic area and distance on $\Hh$:
\begin{equation}\label{eq:hyperbolic}
 dm(z)=\frac{dx\,dy}{y^2},\qquad
 \cosh d_{\Hh}(z,w)=1+\frac{|z-w|^2}{2\Imn z\Imn w},
 \qquad z=x+iy.
\end{equation}
For $t>0$, let $B_t(w)$ be the hyperbolic ball of radius $t$ about
$w=x_0+iy_0$. Its Euclidean description and area are
\begin{equation}\label{eq:ballgeometry}
 \begin{split}
 B_t(w)&=\{x+iy:(x-x_0)^2+(y-y_0\cosh t)^2<y_0^2\sinh^2t\},\\
 V_t&=m(B_t(w))=2\pi(\cosh t-1).
 \end{split}
\end{equation}
These formulas follow from
\cite[Section~2.1, Proposition~2.14 and Lemma~2.16]{SeriesHyperbolic}.
The maps $g(z)=(az+b)/(cz+d)$ with real coefficients and $ad-bc=1$
preserve distance and area
\cite[Theorem~1.27 and Proposition~2.5]{SeriesHyperbolic}.
They form $\operatorname{PSL}_2(\R)=\operatorname{SL}_2(\R)/\{I,-I\}$.
We identify each matrix with its map; matrix multiplication corresponds
to composition, $gh=g\circ h$.
In particular, for real $a$ we write
\[
 g_a(w)=-\frac1{a+w},\qquad
 g_a=\begin{pmatrix}0&-1\\1&a\end{pmatrix}.
\]

For a Borel measure $\mu$, its image under $g$ is
$(g_\#\mu)(B)=\mu(g^{-1}B)$.
For $t>0$, define the averages of a finite positive measure and of a
function, and the operator $T_g$ on functions, by
\begin{equation}\label{eq:smoothing}
 S_t\mu(z)=\frac{\mu(B_t(z))}{V_t},\qquad
 S_tf(z)=\frac1{V_t}\int_{B_t(z)}f(w)\,dm(w),\qquad
 T_gf=f\circ g^{-1}.
\end{equation}
Thus $g$ acts on points, whereas $T_g$ acts on functions.
All $L^a$ norms below use $m$, and $\|\cdot\|_{2\to2}$ denotes the
operator norm on $L^2(\Hh,m)$. We write $fm$ for a measure with
density $f$; then $S_t(fm)=S_tf$ when $fm$ is finite and positive.

The symmetric kernel $K_t(z,w)=V_t^{-1}\one_{B_t(w)}(z)$ has integral
one in either variable. Tonelli and
$\|S_t\mu\|_2^2\leq\|S_t\mu\|_\infty\|S_t\mu\|_1$ give, for a
probability measure $\mu$,
\begin{equation}\label{eq:smoothingbounds}
 \int S_t\mu\,dm=1,\qquad 0\leq S_t\mu\leq V_t^{-1},\qquad
 0<\|S_t\mu\|_2\leq V_t^{-1/2},\qquad \|S_t\|_{2\to2}\leq1.
\end{equation}
The last inequality is Schur's test
\cite[Lemma~0.32, p.~28]{Teschl}, with both kernel factors equal to
$K_t^{1/2}$. Invariance of distance and area gives
\begin{equation}\label{eq:commutation}
 S_t(g_\#\mu)=T_gS_t\mu,\qquad
 \|T_gf\|_2=\|f\|_2,\qquad T_gT_h=T_{gh}.
\end{equation}
These facts apply also to singular measures, including point masses.

\subsection{A quantitative group estimate}
The following consequence of Kesten's free-group estimate supplies
the constant used in the proof. We keep the short reduction to the
specific maps that occur in the resolvent recursion.

\begin{lemma}\label{lem:relativegap}
Let $D\geq2$ be an integer, $0<|E|\leq D$, and $N_D=3D$. Put
\begin{equation*}
 g_E(w)=-\frac1{w+E},\qquad
 h_E(w)=w-E^{-1},\qquad
 \ell_E=g_E^{-1}h_Eg_E,
 \qquad \delta_D=\frac{4-2\sqrt3}{N_D^2}.
\end{equation*}
Then, for every $f\in L^2(\Hh,m)$,
\begin{equation}\label{eq:commongap}
 \norm{T_{h_E}f-f}_2^2+\norm{T_{\ell_E}f-f}_2^2
 \geq\delta_D\norm f_2^2.
\end{equation}
\end{lemma}
\begin{proof}
First consider the matrices
\[
 \mathsf U=\begin{pmatrix}1&t\\0&1\end{pmatrix},\qquad
 \mathsf V=\begin{pmatrix}1&0\\-t&1\end{pmatrix},\qquad |t|\geq3.
\]
The usual two-cone argument proves that they generate a discrete free
group; see also \cite{Sanov,Brenner,LyndonUllman} for this classical
construction. Here the needed separation follows directly: for
$n\in\mathbb Z\setminus\{0\}$,
\[
 |y|>|x|\ \Longrightarrow\ |x+nty|>2|y|,\qquad
 |x|>|y|\ \Longrightarrow\ |y-ntx|>2|x|.
\]
Each nonempty reduced word $W$ with $b$ alternating nonzero-power
blocks therefore sends one of the coordinate unit vectors to a vector
of maximum norm greater than $2^b$. Hence
$\|W\mp I\|_{\infty\to\infty}>1$, proving both freeness and
discreteness in $\operatorname{PSL}_2(\R)$.

Write $\Lambda=\langle\mathsf U,\mathsf V\rangle$.
Its action on $\Hh$ is properly discontinuous
\cite[Theorem~4.20]{SeriesHyperbolic} and free, since point stabilizers
are finite and $\Lambda$ is torsion-free.
Choose a Dirichlet fundamental domain $F$
\cite[Theorem~5.3 and Section~5.2.1]{SeriesHyperbolic}; its translates
partition $\Hh$ up to their area-zero boundaries. The identification
\begin{equation}\label{eq:regularrep}
 L^2(\Hh,m)\simeq L^2(F,m;\ell^2(\Lambda)),\qquad
 f\longmapsto\bigl[z\mapsto(f(\gamma z))_{\gamma\in\Lambda}\bigr]
\end{equation}
sends $T_\sigma$ to left translation
$(\lambda_\sigma u)(\gamma)=u(\sigma^{-1}\gamma)$ in the second
coordinate. Kesten's theorem \cite[Theorem~3, p.~347]{Kesten} gives
$\|\lambda_{\mathsf U}+\lambda_{\mathsf U}^{*}
  +\lambda_{\mathsf V}+\lambda_{\mathsf V}^{*}\|=2\sqrt3$.
Expanding the two squares and integrating over $F$ yields
\begin{equation}\label{eq:freegap}
 \|T_{\mathsf U}f-f\|_2^2+\|T_{\mathsf V}f-f\|_2^2
 \geq(4-2\sqrt3)\|f\|_2^2.
\end{equation}

For our maps, the matrices are
\begin{equation}\label{eq:commonmatrices}
 g_E=\begin{pmatrix}0&-1\\1&E\end{pmatrix},\qquad
 h_E=\begin{pmatrix}1&-E^{-1}\\0&1\end{pmatrix},\qquad
 \ell_E=\begin{pmatrix}0&-E\\E^{-1}&2\end{pmatrix}.
\end{equation}
With $B_E=\left(\begin{smallmatrix}1&E\\0&1\end{smallmatrix}\right)$,
direct multiplication gives
\[
 B_Eh_E^{N_D}B_E^{-1}=\begin{pmatrix}1&-N_D/E\\0&1\end{pmatrix},\qquad
 B_E\ell_E^{N_D}B_E^{-1}=\begin{pmatrix}1&0\\N_D/E&1\end{pmatrix}.
\]
Since $|N_D/E|\geq3$, apply \eqref{eq:freegap} to $T_{B_E}f$ and
use the unitary conjugation identity in \eqref{eq:commutation}.
Finally, for any unitary $V$,
\[
 V^N-I=\sum_{j=0}^{N-1}V^j(V-I),\qquad
 \|(V^N-I)f\|\leq N\|(V-I)f\|.
\]
Applying this to $T_{h_E}$ and $T_{\ell_E}$ divides the bound in
\eqref{eq:freegap} by $N_D^2$ and proves \eqref{eq:commongap}.
\end{proof}
\subsection{Convolution bounds for singular laws}\label{sec:sums}
For probability measures $\mu,\nu$ on $\Hh$, define their convolution
$\mu*\nu$ on a Borel set $B\subset\Hh$ by
\[
 (\mu*\nu)(B)=\int_{\Hh}\int_{\Hh}\one_B(z+w)\,d\mu(z)\,d\nu(w)
 \quad(B\subset\Hh\text{ Borel}).
\]
Thus $\mu*\nu$ is the distribution of the ordinary complex sum of independent
variables with distributions $\mu,\nu$. For integers $j\geq1$, we write
\[
 \nu^{*1}=\nu,\qquad \nu^{*(j+1)}=\nu^{*j}*\nu.
\]
Adding $b\in\Hh$ changes hyperbolic distances: for $z\ne w$,
\[
 \cosh d_{\Hh}(z+b,w+b)-1
 =\frac{|z-w|^2}{2(\Imn z+\Imn b)(\Imn w+\Imn b)}
 <\cosh d_{\Hh}(z,w)-1.
\]
For the estimates in this subsection, fix
\begin{equation}\label{eq:geometryconstants}
 R=\log\frac43,\qquad r=\log\frac54,\qquad L=2,\qquad
 C=\frac{4V_LV_{L+R}}{V_rV_R}>4.
\end{equation}
Set
\begin{equation}\label{eq:Bk}
 B_k=\left(\frac C2\right)^{\lceil\log_2 k\rceil},\qquad k\geq2.
\end{equation}

The next estimate applies even when the distributions have no densities.
For densities, the required bound follows from Young's convolution
inequality. A comparison of ball averages then gives the result for
arbitrary probability measures.

\begin{proposition}\label{prop:convolution}
For probability measures $\mu,\nu$ on $\Hh$,
\begin{equation}\label{eq:mixed}
 \norm{S_R(\mu*\nu)}_2
 \leq\frac C4\bigl(\norm{S_R\mu}_2+\norm{S_R\nu}_2\bigr).
\end{equation}
Consequently, with $B_k$ defined in \eqref{eq:Bk},
\begin{equation}\label{eq:jsum}
 \norm{S_R(\nu^{*j})}_2\leq B_k\norm{S_R\nu}_2,
 \qquad 1\leq j\leq k.
\end{equation}
\end{proposition}
\begin{proof}
First let $\mu=fm$ and $\nu=gm$, where $f,g\in L^2(\Hh,m)$ are
probability densities. Put $F(x,y)=y^{-2}f(x+iy)$ and
$G(x,y)=y^{-2}g(x+iy)$ for $y>0$, and extend both functions by zero
to $\R^2$. Write $dA=dx\,dy$. With ordinary convolution on $\R^2$,
the density of the sum with respect to $m$ is
\begin{equation}\label{eq:convolutiondensity}
 \mathcal Q(f,g)(x+iy)=y^2(F*G)(x,y),\qquad
 y(F*G)=(yF)*G+F*(yG).
\end{equation}
Young's $L^1*L^2\to L^2$ inequality, which in this context also follows from Schur's
test \cite[Lemma~0.32, p.~28]{Teschl} with kernel $G(a-b)$ or $F(a-b)$
on $\R^2$, gives
\begin{equation}\label{eq:densitybound}
 \begin{aligned}
 \|\mathcal Q(f,g)\|_2
 &=\|y(F*G)\|_{L^2(dA)}\leq\|yF\|_{L^2(dA)}\|G\|_{L^1(dA)}
       +\|F\|_{L^1(dA)}\|yG\|_{L^2(dA)}
 =\|f\|_2+\|g\|_2.
 \end{aligned}
\end{equation}
Here $\|F\|_{L^1(dA)}=\|G\|_{L^1(dA)}=1$ because $fm$ and $gm$
are probability measures.

To pass to arbitrary measures, we compare their ball averages.
The ball formula \eqref{eq:ballgeometry} implies
\begin{equation}\label{eq:ballbounds}
 z\in B_t(w)\quad\Longrightarrow\quad
 |z-w|\leq(e^t-1)\Imn w,\qquad
 e^{-t}\Imn w\leq\Imn z\leq e^t\Imn w.
\end{equation}
Set $\kappa_L(z,w)=V_L^{-1}\one_{B_L(w)}(z)$ and let
$h_{w_1,w_2}=\mathcal Q(\kappa_L(\cdot,w_1),\kappa_L(\cdot,w_2))$.
Suppose first that $u:=\Imn w_1\leq v:=\Imn w_2$.
For $z\in B_R(w_1+w_2)$ and $w\in B_r(w_1)$,
\eqref{eq:ballbounds} and the chosen radii give
\[
 \Imn(z-w)\geq\tfrac34(u+v)-\tfrac54u\geq\tfrac v4,
 \qquad
 |z-w-w_2|\leq\tfrac{u+v}{3}+\tfrac u4\leq\tfrac{11v}{12}.
\]
Consequently,
\begin{equation}\label{eq:ballinclusion}
 \cosh d_{\Hh}(z-w,w_2)
 \leq1+\frac{(11v/12)^2}{2v(v/4)}=\frac{193}{72}<\cosh L,
\end{equation}
so $w\in B_L(w_1)$ and $z-w\in B_L(w_2)$.
Writing $y=\Imn z$, the convolution density therefore satisfies
\[
 h_{w_1,w_2}(z)
 =\frac1{V_L^2}\int_{\substack{w\in B_L(w_1)\\z-w\in B_L(w_2)}}
       \frac{y^2}{(\Imn(z-w))^2}\,dm(w)
 \geq\frac{V_r}{V_L^2},\qquad z\in B_R(w_1+w_2).
\]
The last inequality uses $0<\Imn(z-w)<y$ and integrates over
$B_r(w_1)$. Symmetry in $w_1,w_2$ gives the same bound when $u>v$.
Integrating this bound against $\mu(dw_1)\nu(dw_2)$ yields
\begin{equation}\label{eq:positivecomparison}
 S_R(\mu*\nu)\leq\frac{V_L^2}{V_rV_R}
                         \mathcal Q(S_L\mu,S_L\nu).
\end{equation}
Tonelli's theorem applies because all kernels are nonnegative;
no density assumption on $\mu,\nu$ is needed.

It remains to compare the two averaging radii. If $w\in B_L(z)$,
then $B_R(w)\subset B_{L+R}(z)$ by the triangle inequality. Hence
\[
 S_{L+R}(S_R\mu)(z)
 =\frac1{V_{L+R}V_R}
       \int m(B_{L+R}(z)\cap B_R(w))\,d\mu(w)
 \geq\frac{V_L}{V_{L+R}}S_L\mu(z).
\]
The $L^2$ contraction in \eqref{eq:smoothingbounds} now gives
\begin{equation}\label{eq:scalecomparison}
 \|S_L\mu\|_2\leq\frac{V_{L+R}}{V_L}\|S_R\mu\|_2.
\end{equation}
Since $S_L\mu$ and $S_L\nu$ are $L^2(m)$ probability densities,
\eqref{eq:densitybound}, \eqref{eq:positivecomparison} and
\eqref{eq:scalecomparison} imply
\[
 \|S_R(\mu*\nu)\|_2
 \leq\frac{V_LV_{L+R}}{V_rV_R}
          \bigl(\|S_R\mu\|_2+\|S_R\nu\|_2\bigr).
\]
Since $C/4=V_LV_{L+R}/(V_rV_R)$, this proves \eqref{eq:mixed}.

Finally, put $A_n=\max_{1\leq j\leq2^n}\|S_R(\nu^{*j})\|_2$.
For $2\leq j\leq2^n$, split the sum into
$\lfloor j/2\rfloor$ and $\lceil j/2\rceil$ terms.
Equation~\eqref{eq:mixed} gives
\[
 A_n\leq\frac C2 A_{n-1}.
\]
Since $C>4$, this also covers $j=1$. Hence
$A_n\leq(C/2)^n\|S_R\nu\|_2$; taking
$n=\lceil\log_2 k\rceil$ proves \eqref{eq:jsum}.
\end{proof}

\subsection{A fixed-point exclusion criterion}\label{sec:fixedpoint-exclusion}
The following observation separates the final functional-analytic step
from the probabilistic calculation. 

\begin{proposition}\label{prop:fixedpoint-exclusion}
Let $\zeta$ be a countably supported probability measure on
$\operatorname{PSL}_2(\R)$. Define its averages on measures and functions by
\[
 P\mu=\int g_\#\mu\,d\zeta(g),\qquad
 Pf=\int T_gf\,d\zeta(g).
\]
Let $Q$ map probability measures on $\Hh$ to probability measures on
$\Hh$; no linearity is assumed for $Q$. Suppose there are constants
$0\leq\rho<1$ and $B<\infty$ such that
\[
 \|P^2\|_{2\to2}\leq\rho,\qquad
 \|S_R(Q\mu)\|_2\leq B\|S_R\mu\|_2
 \quad\text{for every probability measure }\mu\text{ on }\Hh.
\]
For $q\in[0,1]$, the equation $\nu=(1-q)P\nu+qQ\nu$ has no probability
solution on $\Hh$ whenever
\begin{equation}\label{eq:fixedpoint-threshold}
 (1-q)^2\rho+q(2-q)B<1.
\end{equation}
In particular, the sufficient condition $\rho+2qB<1$ implies
\eqref{eq:fixedpoint-threshold}.
\end{proposition}
\begin{proof}
Suppose that $\nu$ is a probability solution, and set
$f=S_R\nu$ and $h=S_R(Q\nu)$. By \eqref{eq:smoothingbounds},
$0<\|f\|_2<\infty$. The commutation relation
\eqref{eq:commutation} gives $S_R(P\mu)=P(S_R\mu)$, and averaging
unitaries gives $\|P\|_{2\to2}\leq1$. Linearity of $P$, without any
iteration of $Q$, gives
\begin{equation}\label{eq:twostepfixedpoint}
 \nu=(1-q)^2P^2\nu+q(1-q)P(Q\nu)+qQ\nu.
\end{equation}
Applying $S_R$ to \eqref{eq:twostepfixedpoint}, we obtain
\[
 f=(1-q)^2P^2f+q(1-q)Ph+qh.
\]
Consequently,
\[
 \|f\|_2\leq(1-q)^2\rho\|f\|_2+q(2-q)\|h\|_2
 \leq\bigl((1-q)^2\rho+q(2-q)B\bigr)\|f\|_2.
\]
This contradicts \eqref{eq:fixedpoint-threshold} and $\|f\|_2>0$.
\end{proof}

\section{The regular-tree argument}\label{sec:trees}

\subsection{The infinite forward cluster and its finite side branches}\label{sec:resolvent-setup}
We fix $k\geq2$ and $p\in(1/k,1)$, and write
$p_\ac(k)=p_\ac(\Tt_{k+1})$. We omit the subscript $p$ on probabilities
and expectations when the parameter is fixed.

Choose a fixed neighbor $o_-$ of $o$ in $\Tt_{k+1}$. Delete the single
edge $\{o_-,o\}$ from this deterministic tree. The deletion separates the
tree into two components. Discard the entire component containing $o_-$,
and call the component containing $o$ the \emph{forward tree} $\Tt^+$.
This cut is imposed regardless of whether $\omega_{\{o_-,o\}}$ is zero
or one; it is not a percolation event.

Root $\Tt^+$ at $o$ and orient each edge away from $o$. The root has
$k$ children and no parent: one of its original $k+1$ neighbors has been
removed. Every other vertex has one parent and $k$ children. These
arrows only specify the parent--child relations; the adjacency operators
still use undirected edges.

Now restrict the configuration $\omega$ to $\Tt^+$, retain its open
edges, and let $T$ be the connected component containing $o$. Equivalently,
using the path notation $[o,x]$ introduced above,
\[
 V(T)=\{x\in V(\Tt^+):\omega_e=1\text{ for every }e\in[o,x]\}.
\]
Every vertex of $\Tt^+$ has exactly $k$ children. After percolation,
 a vertex of $T$ has between zero and $k$ children, and $T$ can be
 finite or infinite. Figure~\ref{fig:forward-construction}
shows the construction for $k=2$.

\begin{figure}[htbp]
\centering
\begingroup
\begin{tikzpicture}[
  x=1cm,y=0.95cm,
  font=\fontsize{9}{10.5}\selectfont,
  line cap=round,line join=round,
  fwd/title/.style={font=\fontsize{9}{11}\selectfont\bfseries,align=center},
  fwd/note/.style={align=center,inner sep=0pt},
  fwd/vertex/.style={circle,fill=black,draw=black,inner sep=0pt,minimum size=3.5pt},
  fwd/outside vertex/.style={circle,fill=white,draw=black!60,line width=0.55pt,
                            inner sep=0pt,minimum size=3.5pt},
  fwd/tree edge/.style={draw=black,line width=0.60pt},
  fwd/oriented/.style={draw=black,line width=0.60pt,->,>=stealth,shorten >=2.5pt},
  fwd/cluster edge/.style={draw=black,line width=1.7pt},
  fwd/outside edge/.style={draw=black!60,line width=0.60pt},
  fwd/closed/.style={draw=black!60,line width=0.70pt,dash pattern=on 2.3pt off 1.7pt}
]
\node[font=\fontsize{10}{12}\selectfont\bfseries] at (5.40,3.92)
  {From the regular tree to its open forward cluster ($k=2$)};
\draw[black!20,line width=0.4pt] (2.70,3.50)--(2.70,-3.05);
\draw[black!20,line width=0.4pt] (8.10,3.50)--(8.10,-3.05);

\begin{scope}
  \node[fwd/title] at (0,3.40) {1. Cut one fixed edge};
  \node at (0,3.00) {The original tree $\mathbb T_3$};
  \path[fill=black!5,rounded corners=3pt] (-2.25,0.85) rectangle (2.25,2.72);
  \node[font=\fontsize{8.5}{10}\selectfont] at (0,2.51) {discard this whole component};
  \coordinate (parent) at (0,1.10);
  \draw[black!60,line width=0.6pt] (parent)--(-0.70,1.72);
  \draw[black!60,line width=0.6pt] (parent)--(0.70,1.72);
  \foreach \x in {-0.70,0.70}{
    \draw[black!60,line width=0.6pt] (\x,1.72)--(\x-0.22,2.04);
    \draw[black!60,line width=0.6pt] (\x,1.72)--(\x+0.22,2.04);
    \node[circle,fill=black!60,inner sep=0pt,minimum size=3.5pt] at (\x,1.72) {};
    \node[text=black!60,font=\scriptsize] at (\x,2.21) {$\vdots$};
  }
  \node[circle,fill=black!60,inner sep=0pt,minimum size=3.5pt] at (parent) {};
  \node[anchor=west,inner sep=1pt] at (0.12,1.13) {$o_-$};
  \draw[fwd/tree edge] (0,0)--(0,0.42);
  \draw[black!60,line width=0.6pt] (0,0.71)--(parent);
  \draw[line width=0.8pt] (-0.11,0.46)--(0.11,0.67);
  \draw[line width=0.8pt] (-0.11,0.67)--(0.11,0.46);
  \node[anchor=west,inner sep=1pt] at (0.19,0.565) {delete $\{o_-,o\}$};
  \coordinate (o) at (0.00,0.00);
  \coordinate (a) at (-1.15,-0.95);
  \coordinate (b) at (1.15,-0.95);
  \coordinate (aa) at (-1.72,-1.90);
  \coordinate (ab) at (-0.57,-1.90);
  \coordinate (ba) at (0.57,-1.90);
  \coordinate (bb) at (1.72,-1.90);
  \coordinate (aaa) at (-1.91,-2.47);
  \coordinate (aab) at (-1.53,-2.47);
  \coordinate (aba) at (-0.76,-2.47);
  \coordinate (abb) at (-0.38,-2.47);
  \coordinate (baa) at (0.38,-2.47);
  \coordinate (bab) at (0.76,-2.47);
  \coordinate (bba) at (1.53,-2.47);
  \coordinate (bbb) at (1.91,-2.47);
  \draw[fwd/tree edge] (o)--(a);
  \draw[fwd/tree edge] (o)--(b);
  \draw[fwd/tree edge] (a)--(aa);
  \draw[fwd/tree edge] (a)--(ab);
  \draw[fwd/tree edge] (b)--(ba);
  \draw[fwd/tree edge] (b)--(bb);
  \draw[fwd/tree edge] (aa)--(aaa);
  \draw[fwd/tree edge] (aa)--(aab);
  \draw[fwd/tree edge] (ab)--(aba);
  \draw[fwd/tree edge] (ab)--(abb);
  \draw[fwd/tree edge] (ba)--(baa);
  \draw[fwd/tree edge] (ba)--(bab);
  \draw[fwd/tree edge] (bb)--(bba);
  \draw[fwd/tree edge] (bb)--(bbb);
  \node[fwd/vertex] at (o) {};
  \node[fwd/vertex] at (a) {};
  \node[fwd/vertex] at (b) {};
  \node[fwd/vertex] at (aa) {};
  \node[fwd/vertex] at (ab) {};
  \node[fwd/vertex] at (ba) {};
  \node[fwd/vertex] at (bb) {};
  \node[fwd/vertex] at (aaa) {};
  \node[fwd/vertex] at (aab) {};
  \node[fwd/vertex] at (aba) {};
  \node[fwd/vertex] at (abb) {};
  \node[fwd/vertex] at (baa) {};
  \node[fwd/vertex] at (bab) {};
  \node[fwd/vertex] at (bba) {};
  \node[fwd/vertex] at (bbb) {};
  \node[text=black,font=\scriptsize] at (-1.91,-2.82) {$\vdots$};
  \node[text=black,font=\scriptsize] at (-1.53,-2.82) {$\vdots$};
  \node[text=black,font=\scriptsize] at (-0.76,-2.82) {$\vdots$};
  \node[text=black,font=\scriptsize] at (-0.38,-2.82) {$\vdots$};
  \node[text=black,font=\scriptsize] at (0.38,-2.82) {$\vdots$};
  \node[text=black,font=\scriptsize] at (0.76,-2.82) {$\vdots$};
  \node[text=black,font=\scriptsize] at (1.53,-2.82) {$\vdots$};
  \node[text=black,font=\scriptsize] at (1.91,-2.82) {$\vdots$};
  \node[anchor=west,inner sep=1pt] at (0.10,0.04) {$o$};
\end{scope}

\begin{scope}[xshift=5.40cm]
  \node[fwd/title] at (0,3.40) {2. Keep the side of $o$};
  \node at (0,3.00) {$\mathbb T^+$: deterministic};
  \node[fwd/note] at (0,2.18) {Every vertex has two children.};
  \node[fwd/note] at (0,1.66) {Only the root has no parent.};
  \node[fwd/note] at (0,1.14) {Arrows point away from $o$.};
  \coordinate (o) at (0.00,0.00);
  \coordinate (a) at (-1.15,-0.95);
  \coordinate (b) at (1.15,-0.95);
  \coordinate (aa) at (-1.72,-1.90);
  \coordinate (ab) at (-0.57,-1.90);
  \coordinate (ba) at (0.57,-1.90);
  \coordinate (bb) at (1.72,-1.90);
  \coordinate (aaa) at (-1.91,-2.47);
  \coordinate (aab) at (-1.53,-2.47);
  \coordinate (aba) at (-0.76,-2.47);
  \coordinate (abb) at (-0.38,-2.47);
  \coordinate (baa) at (0.38,-2.47);
  \coordinate (bab) at (0.76,-2.47);
  \coordinate (bba) at (1.53,-2.47);
  \coordinate (bbb) at (1.91,-2.47);
  \draw[fwd/oriented] (o)--(a);
  \draw[fwd/oriented] (o)--(b);
  \draw[fwd/oriented] (a)--(aa);
  \draw[fwd/oriented] (a)--(ab);
  \draw[fwd/oriented] (b)--(ba);
  \draw[fwd/oriented] (b)--(bb);
  \draw[fwd/oriented] (aa)--(aaa);
  \draw[fwd/oriented] (aa)--(aab);
  \draw[fwd/oriented] (ab)--(aba);
  \draw[fwd/oriented] (ab)--(abb);
  \draw[fwd/oriented] (ba)--(baa);
  \draw[fwd/oriented] (ba)--(bab);
  \draw[fwd/oriented] (bb)--(bba);
  \draw[fwd/oriented] (bb)--(bbb);
  \node[fwd/vertex] at (o) {};
  \node[fwd/vertex] at (a) {};
  \node[fwd/vertex] at (b) {};
  \node[fwd/vertex] at (aa) {};
  \node[fwd/vertex] at (ab) {};
  \node[fwd/vertex] at (ba) {};
  \node[fwd/vertex] at (bb) {};
  \node[fwd/vertex] at (aaa) {};
  \node[fwd/vertex] at (aab) {};
  \node[fwd/vertex] at (aba) {};
  \node[fwd/vertex] at (abb) {};
  \node[fwd/vertex] at (baa) {};
  \node[fwd/vertex] at (bab) {};
  \node[fwd/vertex] at (bba) {};
  \node[fwd/vertex] at (bbb) {};
  \node[text=black,font=\scriptsize] at (-1.91,-2.82) {$\vdots$};
  \node[text=black,font=\scriptsize] at (-1.53,-2.82) {$\vdots$};
  \node[text=black,font=\scriptsize] at (-0.76,-2.82) {$\vdots$};
  \node[text=black,font=\scriptsize] at (-0.38,-2.82) {$\vdots$};
  \node[text=black,font=\scriptsize] at (0.38,-2.82) {$\vdots$};
  \node[text=black,font=\scriptsize] at (0.76,-2.82) {$\vdots$};
  \node[text=black,font=\scriptsize] at (1.53,-2.82) {$\vdots$};
  \node[text=black,font=\scriptsize] at (1.91,-2.82) {$\vdots$};
  \node[anchor=west,inner sep=1pt] at (0.10,0.04) {$o$};
\end{scope}

\begin{scope}[xshift=10.80cm]
  \node[fwd/title] at (0,3.40) {3. Select the open cluster};
  \node at (0,3.00) {$T$: random};
  \node[fwd/note] at (0,2.25) {Open probability: $p$.};
  \node[fwd/note] at (0,1.80) {Edge states are independent.};
  \node[fwd/note] at (0,1.18) {Bold = the cluster of $o$.};
  \node[fwd/note] at (0,0.73) {Here $T$ has three vertices.};
  \coordinate (o) at (0.00,0.00);
  \coordinate (a) at (-1.15,-0.95);
  \coordinate (b) at (1.15,-0.95);
  \coordinate (aa) at (-1.72,-1.90);
  \coordinate (ab) at (-0.57,-1.90);
  \coordinate (ba) at (0.57,-1.90);
  \coordinate (bb) at (1.72,-1.90);
  \coordinate (aaa) at (-1.91,-2.47);
  \coordinate (aab) at (-1.53,-2.47);
  \coordinate (aba) at (-0.76,-2.47);
  \coordinate (abb) at (-0.38,-2.47);
  \coordinate (baa) at (0.38,-2.47);
  \coordinate (bab) at (0.76,-2.47);
  \coordinate (bba) at (1.53,-2.47);
  \coordinate (bbb) at (1.91,-2.47);
  \draw[fwd/cluster edge] (o)--(a);
  \draw[fwd/closed] (o)--(b);
  \draw[fwd/cluster edge] (a)--(aa);
  \draw[fwd/closed] (a)--(ab);
  \draw[fwd/outside edge] (b)--(ba);
  \draw[fwd/outside edge] (b)--(bb);
  \draw[fwd/closed] (aa)--(aaa);
  \draw[fwd/closed] (aa)--(aab);
  \draw[fwd/outside edge] (ab)--(aba);
  \draw[fwd/closed] (ab)--(abb);
  \draw[fwd/outside edge] (ba)--(baa);
  \draw[fwd/closed] (ba)--(bab);
  \draw[fwd/closed] (bb)--(bba);
  \draw[fwd/outside edge] (bb)--(bbb);
  \node[fwd/vertex] at (o) {};
  \node[fwd/vertex] at (a) {};
  \node[fwd/outside vertex] at (b) {};
  \node[fwd/vertex] at (aa) {};
  \node[fwd/outside vertex] at (ab) {};
  \node[fwd/outside vertex] at (ba) {};
  \node[fwd/outside vertex] at (bb) {};
  \node[fwd/outside vertex] at (aaa) {};
  \node[fwd/outside vertex] at (aab) {};
  \node[fwd/outside vertex] at (aba) {};
  \node[fwd/outside vertex] at (abb) {};
  \node[fwd/outside vertex] at (baa) {};
  \node[fwd/outside vertex] at (bab) {};
  \node[fwd/outside vertex] at (bba) {};
  \node[fwd/outside vertex] at (bbb) {};
  \node[text=black!60,font=\scriptsize] at (-1.91,-2.82) {$\vdots$};
  \node[text=black!60,font=\scriptsize] at (-1.53,-2.82) {$\vdots$};
  \node[text=black!60,font=\scriptsize] at (-0.76,-2.82) {$\vdots$};
  \node[text=black!60,font=\scriptsize] at (-0.38,-2.82) {$\vdots$};
  \node[text=black!60,font=\scriptsize] at (0.38,-2.82) {$\vdots$};
  \node[text=black!60,font=\scriptsize] at (0.76,-2.82) {$\vdots$};
  \node[text=black!60,font=\scriptsize] at (1.53,-2.82) {$\vdots$};
  \node[text=black!60,font=\scriptsize] at (1.91,-2.82) {$\vdots$};
  \node[anchor=west,inner sep=1pt] at (0.10,0.04) {$o$};
\end{scope}

\draw[black!25,line width=0.4pt] (-2.30,-3.15)--(13.10,-3.15);
\node[anchor=east,font=\fontsize{8.5}{10}\selectfont] at (-1.60,-3.50) {(3):};
\draw[fwd/cluster edge] (-1.35,-3.50)--(-0.65,-3.50);
\node[fwd/vertex] at (-1.35,-3.50) {};
\node[fwd/vertex] at (-0.65,-3.50) {};
\node[anchor=west] at (-0.48,-3.50) {open, in $T$};
\draw[fwd/outside edge] (3.55,-3.50)--(4.25,-3.50);
\node[fwd/outside vertex] at (3.55,-3.50) {};
\node[fwd/outside vertex] at (4.25,-3.50) {};
\node[anchor=west] at (4.42,-3.50) {open, outside $T$};
\draw[fwd/closed] (9.80,-3.50)--(10.50,-3.50);
\node[anchor=west] at (10.67,-3.50) {closed};
\node[align=center,font=\fontsize{9}{11}\selectfont] at (5.40,-4.36)
  {An open edge can lie outside $T$: a closed edge may block its path to $o$.\\[3pt]
   Dots denote omitted generations of the ambient tree, not leaves.};
\end{tikzpicture}
\endgroup

\caption{From the regular tree to the open forward cluster, for $k=2$.
The first panel cuts the fixed edge $\{o_-,o\}$ and discards the entire
component above it. The second panel is the deterministic tree $\Tt^+$,
with two children at every vertex. The third panel shows one percolation
configuration: bold edges and filled vertices form $T$, while thin solid
edges are open but disconnected from $o$. The illustrated cluster has
three vertices because all four of its boundary edges are closed.
The definition of $T$ does not condition on survival. Dots mark omitted
generations of the ambient tree, not leaves.}
\label{fig:forward-construction}
\end{figure}
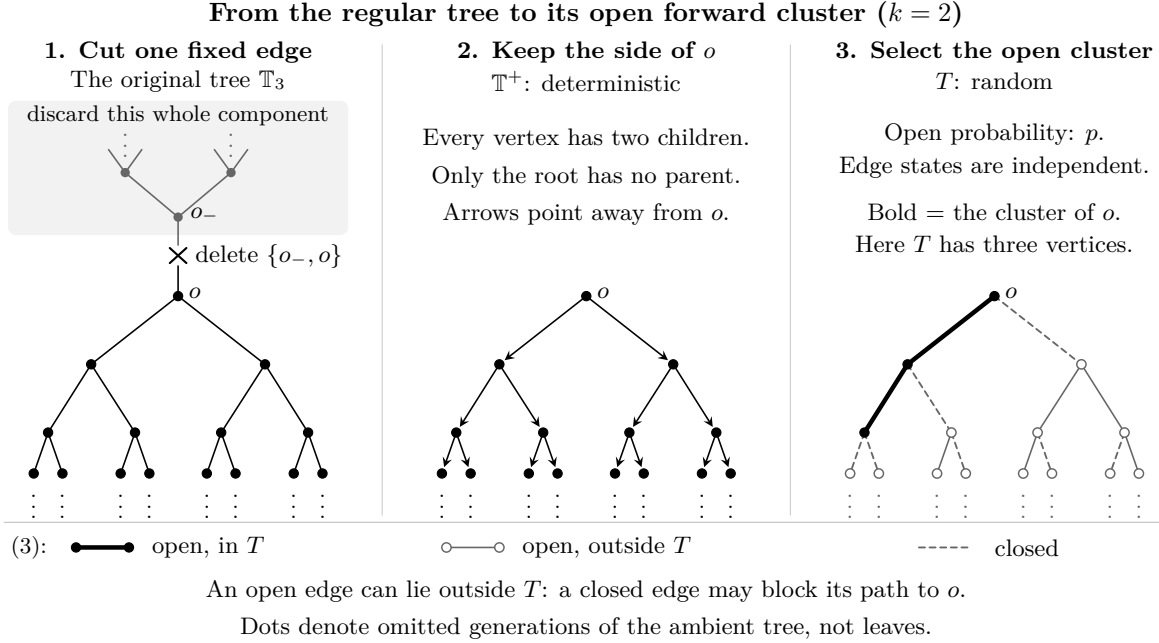

Let $d(x,y)$ be the number of edges in the unique path from $x$
to $y$ in the deterministic tree.
For $x\in V(\Tt^+)$, define its child set by
\[
 \operatorname{Ch}(x)=\{y\in V(\Tt^+):y\sim x,\ d(o,y)=d(o,x)+1\},
 \qquad |\operatorname{Ch}(x)|=k.
\]
The original root cluster $C_o$ from \eqref{eq:cluster} may also use the
edge $\{o_-,o\}$. In particular,
\[
 V(T)\subseteq V(C_o),\qquad
 \#\{y:y\sim o\}=k+1,\qquad |\operatorname{Ch}(o)|=k.
\]
Figure~\ref{fig:forward-construction} shows the passage from the original
percolation model to $\Tt^+$ and $T$.

For each integer $n\geq0$, define the event $\mathcal R_n$, its
probability $s_n$, and the survival probability $s$ by
\[
 \mathcal R_n=\{\exists\,x\in V(T):d(o,x)=n\},\qquad
 s_n=\PP_p(\mathcal R_n),\qquad s=\PP_p(|T|=\infty).
\]
Since each generation is finite,
\[
 \mathcal R_{n+1}\subseteq\mathcal R_n,\qquad
 \bigcap_{n\geq0}\mathcal R_n=\{|T|=\infty\},\qquad
 s_n\downarrow s.
\]
Label the children of $o$ by $o_1,\ldots,o_k$. For $1\leq i\leq k$,
set $\omega_i=\omega_{oo_i}$ and let $T_i$ be the open cluster rooted
at $o_i$ after deleting $\{o,o_i\}$, regardless of $\omega_i$.
Let $\mathcal R_n^{(i)}$ be the event that $T_i$ reaches distance $n$
below $o_i$. Then
\[
 \mathcal R_{n+1}
 =\bigcup_{i=1}^k\bigl(\{\omega_i=1\}\cap\mathcal R_n^{(i)}\bigr),
 \qquad
 \PP_p(\omega_i=1,\mathcal R_n^{(i)})=ps_n.
\]
The $k$ branch events are independent, so $1-s_{n+1}=\prod_{i=1}^k(1-ps_n)=(1-ps_n)^k.$
Define $\varphi(u)=1-(1-pu)^k$ for $u\in[0,1]$, and let
$\eta=1-ps$ be the probability that a given child position does not
produce an infinite open branch. Then
\begin{equation}\label{eq:forwardinfinite}
 \begin{gathered}
 s_0=1,\qquad s_{n+1}=\varphi(s_n),\qquad s_n\downarrow s,\\
 s=1-(1-ps)^k,\qquad \eta=1-ps,\qquad 1-s=\eta^k.
 \end{gathered}
\end{equation}
The forward offspring distribution is $\Bin(k,p)$, with mean $kp$.
The Galton--Watson extinction criterion
\cite[Chapter~5]{LyonsPeres}, also recalled in
\cite[Section~2.2]{Bordenave}, therefore gives
\[
 s>0\quad\Longleftrightarrow\quad kp>1,\qquad
 p_c(\Tt_{k+1})=\frac1k.
\]
Returning to $p\in(1/k,1)$, each of the $k+1$ directions from the root fails
to produce an infinite branch with probability $\eta$. Independence gives
\begin{equation}\label{eq:fullinfinite}
 \PP_p(|C_o|<\infty)=\prod_{i=1}^{k+1}\eta=\eta^{k+1}\text{ and therefore }
 \PP_p(|C_o|=\infty)=1-\eta^{k+1}>0.
\end{equation}
For a possible child with edge state $\omega_i$ and forward cluster
$T_i$, the two ways to avoid an infinite branch are disjoint:
\[
 \begin{aligned}
 \eta
 &=\PP_p(\omega_i=0)+\PP_p(\omega_i=1,\ |T_i|<\infty)=(1-p)+p(1-s)=1-ps.
 \end{aligned}
\]
Bordenave's skeleton decomposition separates a supercritical
Galton--Watson tree into infinite descendants and extinct finite branches
\cite[Section~2.2]{Bordenave}. More precisely, his equation~(7) gives the
joint law, conditioned on survival, of the numbers of infinite and extinct
open child branches; equation~(8) gives the offspring law inside an extinct
branch. His equations~(9)--(10) then insert this decomposition into the root
resolvent recursion. We use the same decomposition for the binomial
percolation tree, but keep all $k$ deterministic child positions: a closed
position is represented by an empty finite branch. The next lemma derives
this specialization directly and computes the two finite configurations
used later.

\begin{lemma}\label{lem:conditioning}
Let $K$ count the infinite open child branches of the root. Then
\begin{equation*}
 \pi_j:=\PP(K=j\mid |T|=\infty)
 =\frac{\binom{k}{j}(ps)^j\eta^{k-j}}{s},\qquad 1\leq j\leq k,
\end{equation*}
\[
 \E[K\mid |T|=\infty]=kp.
\]
\begin{equation}\label{eq:qbound2}
 0<q:=\PP(K\geq2\mid |T|=\infty)
 =\sum_{j=2}^k\pi_j\leq kp-1.
\end{equation}
For a child position $i$, set
\[
 B_i=\begin{cases}T_i,&\omega_i=1,\\ \varnothing,&\omega_i=0,\end{cases}
 \qquad
 \mathcal L(F)=\mathcal L(B_i\mid |B_i|<\infty).
\]
Given $K=j$, relabel the branches so that the infinite ones come first.
Then
\[
 \mathcal L(B'_1,\ldots,B'_k\mid K=j)
 =\mathcal L(T\mid |T|=\infty)^{\otimes j}
  \otimes\mathcal L(F)^{\otimes(k-j)},
\]
and
\begin{equation*}
 \PP(F=\varnothing)=\frac{1-p}{\eta},\qquad
 \PP(|F|=1)=\frac{p(1-p)^k}{\eta}.
\end{equation*}
\end{lemma}

\begin{proof}
We derive the binomial specialization directly. With $B_i$ as in the
statement, set
\[
 X_i:=\mathbf1_{\{|B_i|=\infty\}},\qquad 1\leq i\leq k.
\]
Independence of the child branches gives
\[
 \begin{gathered}
 \PP(X_i=1)=ps,\qquad K=\sum_{i=1}^kX_i\sim\Bin(k,ps),\\
 \mathcal S:=\{|T|=\infty\}=\{K\geq1\},\qquad \PP(\mathcal S)=s.
 \end{gathered}
\]
Hence, for $1\leq j\leq k$,
\[
 \begin{gathered}
 \pi_j=\frac{\PP(K=j)}s
       =\frac{\binom{k}{j}(ps)^j\eta^{k-j}}s,
 \qquad \E[K\mid\mathcal S]=\frac{kps}{s}=kp,\\
 0<\pi_2\leq q=\E[\mathbf1_{\{K\geq2\}}\mid\mathcal S]
       \leq\E[K-1\mid\mathcal S]=kp-1.
 \end{gathered}
\]
For $x_1,\ldots,x_k\in\{0,1\}$, independence also gives
\[
 \mathcal L(B_1,\ldots,B_k\mid X_1=x_1,\ldots,X_k=x_k)
 =\bigotimes_{i=1}^k\mathcal L(B_i\mid X_i=x_i).
\]
Relabel the branches as $B'_1,\ldots,B'_k$, listing infinite positions
first and preserving the order within each group. The product above
then depends only on $j=\sum_i x_i$, so
\[
 \mathcal L(B'_1,\ldots,B'_k\mid K=j)
 =\mathcal L(T\mid |T|=\infty)^{\otimes j}
   \otimes\mathcal L(F)^{\otimes(k-j)},
 \qquad \mathcal L(F)=\mathcal L(B_i\mid X_i=0).
\]
Finally,
\begin{align*}
 \PP(F=\varnothing)
 &=\frac{\PP(\omega_i=0)}{\PP(X_i=0)}=\frac{1-p}{\eta},\\
 \PP(|F|=1)
 &=\frac{\PP(\omega_i=1,\ |T_i|=1)}{\PP(X_i=0)}
   =\frac{p(1-p)^k}{\eta}.\qedhere
\end{align*}
\end{proof}

\subsection{The root resolvent and its conditional distribution}\label{sec:grushin-law}
By \eqref{eq:model}, every operator considered here has spectrum in
\begin{equation}\label{eq:energyinterval}
 I_k=[-(k+1),k+1].
\end{equation}
For a rooted tree $X$ of maximum degree at most $k+1$, with root $o_X$
and adjacency operator $A_X$, define its diagonal resolvent by
\[
 \Gamma_X(z)=\langle\delta_{o_X},(A_X-z)^{-1}\delta_{o_X}\rangle,
 \qquad z\in\Hh.
\]
For the forward cluster, this reads
\[ \Gamma_T(z)=\langle\delta_o,(A_T-z)^{-1}\delta_o\rangle,
 \qquad z\in\Hh.
\]
To check measurability on a fixed space, extend $A_T$ by zero outside
 $V(T)$ in the deterministic forward tree. Thus, on $\ell^2(V(\Tt^+))$,
\[
 \widetilde A_T=A_T\oplus0,\qquad
 \langle\delta_o,(\widetilde A_T-z)^{-1}\delta_o\rangle
 =\Gamma_T(z).
\]
Whether a vertex belongs to the cluster is measurable by the
 open-path formula in Section~\ref{sec:notation}. The extended
 operator therefore has measurable matrix entries. We use the same convention for all forward clusters when
using the boundary-value argument in Section~\ref{sec:boundaryfacts}.

To isolate the value at the root, define the restriction and extension
maps
\[
 R_{+,X}:\ell^2(V(X))\longrightarrow\mathbb C,\quad R_{+,X}u=u(o_X),
 \qquad
 R_{-,X}:\mathbb C\longrightarrow\ell^2(V(X)),\quad R_{-,X}c=c\delta_{o_X}.
\]
We add one scalar unknown and one equation fixing the value at the
 root. This enlarged linear system is called a \emph{Grushin problem}.
 Its operator is
\begin{equation}\label{eq:grushin-root}
 \mathcal P_X(z)=
 \begin{pmatrix}A_X-z&R_{-,X}\\ R_{+,X}&0\end{pmatrix}
 :\ell^2(V(X))\oplus\mathbb C\longrightarrow\ell^2(V(X))\oplus\mathbb C.
\end{equation}
For $z\in\Hh$ this operator is invertible. Write $E_{-+,X}(z)$
 for the scalar in the lower-right corner of its inverse. This scalar
 is usually called the \emph{effective Hamiltonian}; see
 \cite[Section~1]{SjostrandZworski}.

We now compute the inverse of \eqref{eq:grushin-root} for any
nonempty rooted tree $X$ of maximum degree at most $k+1$. Let $o_X$ be
its root. Removing $o_X$ leaves $n$ components $X_1,\ldots,X_n$, where
$0\leq n\leq k+1$; let $o_i$ be the vertex of $X_i$ adjacent to $o_X$.
We regard $X_i$ as rooted at $o_i$. Separate the root coordinate from the coordinates on the branches:
\[
 \mathscr H_X=\ell^2(V(X)),\qquad
 \mathscr K_X=\bigoplus_{i=1}^n\ell^2(V(X_i)),\qquad
 \mathscr H_X=\mathbb C\delta_{o_X}\oplus\mathscr K_X
              \simeq\mathbb C\oplus\mathscr K_X.
\]
In this decomposition, write $u=(u_0,u^\circ)$, with $u_0=u(o_X)$ and
$u^\circ$ the collection of its restrictions to the branches. Define the
branch adjacency operator $B_X$ and root coupling $b_X$ by
\[
 B_X=\bigoplus_{i=1}^n A_{X_i},\qquad
 b_X:\mathbb C\longrightarrow\mathscr K_X,\quad
 b_Xc=(c\delta_{o_i})_{i=1}^n.
\]
For $u^\circ=(u_1,\ldots,u_n)\in\mathscr K_X$,
\[
 b_X^*u^\circ=\sum_{i=1}^n u_i(o_i),\qquad
 \|b_X\|=\sqrt n,\qquad
 A_X=\begin{pmatrix}0&b_X^*\\ b_X&B_X\end{pmatrix}.
\]
For $z\in\Hh$, define the branch resolvent
\[
 \mathcal R_X(z)=(B_X-z)^{-1}
       =\bigoplus_{i=1}^n(A_{X_i}-z)^{-1},\qquad
 \|\mathcal R_X(z)\|\leq(\Imn z)^{-1}.
\]
If $n=0$, then $\mathscr K_X=\{0\}$ and all branch blocks below are zero.

With the root maps $R_{\pm,X}$ defined above, the Grushin
operator acts on $\mathbb C\oplus\mathscr K_X\oplus\mathbb C$ as
\[
 \mathcal P_X(z)=
 \begin{pmatrix}
  -z&b_X^*&1\\
  b_X&B_X-z&0\\
  1&0&0
 \end{pmatrix}.
\]
To invert this operator, let $f=(f_0,f^\circ)\in\mathscr H_X$ and
$f_+\in\mathbb C$ be the prescribed data, and let
$(u,u_-)\in\mathscr H_X\oplus\mathbb C$ be the unknown. The equation
$\mathcal P_X(z)(u,u_-)=(f,f_+)$ reads
\[
 -zu_0+b_X^*u^\circ+u_-=f_0,\qquad
 b_Xu_0+(B_X-z)u^\circ=f^\circ,\qquad u_0=f_+.
\]
These equations have the unique solution
\[
 \begin{aligned}
  u_0&=f_+,\\
  u^\circ&=\mathcal R_X(z)f^\circ-\mathcal R_X(z)b_Xf_+,\\
  u_-&=f_0-b_X^*\mathcal R_X(z)f^\circ
                 +\bigl(z+b_X^*\mathcal R_X(z)b_X\bigr)f_+.
 \end{aligned}
\]
Therefore the lower-right scalar block of the inverse is
\begin{equation}\label{eq:grushin-effective}
 E_{-+,X}(z)=z+b_X^*\mathcal R_X(z)b_X
           =z+\sum_{i=1}^n\Gamma_{X_i}(z),
\end{equation}
and the full inverse is
\begin{equation}\label{eq:grushin-inverse}
 \mathcal P_X(z)^{-1}=
 \begin{pmatrix}
  0&0&1\\
  0&\mathcal R_X(z)&-\mathcal R_X(z)b_X\\
  1&-b_X^*\mathcal R_X(z)&E_{-+,X}(z)
 \end{pmatrix}.
\end{equation}
These formulas give a unique solution, and every block is bounded.
 Thus the enlarged operator has a bounded inverse even when some
 branches are infinite. Moreover, the scalar resolvents have positive imaginary
part, so
\[
 \Imn E_{-+,X}(z)=\Imn z+\sum_{i=1}^n\Imn\Gamma_{X_i}(z)>0.
\]
In particular, $E_{-+,X}(z)\ne0$ for $z\in\Hh$.

Write the remaining blocks of the inverse, with respect to
$\mathscr H_X\oplus\mathbb C$, by
$E_X(z):\mathscr H_X\to\mathscr H_X$,
$E_{+,X}(z):\mathbb C\to\mathscr H_X$, and
$E_{-,X}(z):\mathscr H_X\to\mathbb C$. Equation~\eqref{eq:grushin-inverse}
gives
\[
 E_X(z)=\begin{pmatrix}0&0\\0&\mathcal R_X(z)\end{pmatrix},\qquad
 E_{+,X}(z)=\begin{pmatrix}1\\-\mathcal R_X(z)b_X\end{pmatrix},\qquad
 E_{-,X}(z)=\begin{pmatrix}1&-b_X^*\mathcal R_X(z)\end{pmatrix}.
\]
The inverse equations are
\[
 u=E_X(z)f+E_{+,X}(z)f_+,\qquad
 u_-=E_{-,X}(z)f+E_{-+,X}(z)f_+.
\]
To recover $(A_X-z)u=f$, require $u_-=0$ and choose
$f_+=-E_{-+,X}(z)^{-1}E_{-,X}(z)f$. This proves the Grushin resolvent
formula \cite[Section~1, equation~(1.1)]{SjostrandZworski}:
\begin{equation}\label{eq:grushin-recovery}
 \begin{aligned}
 (A_X-z)^{-1}
   &=E_X(z)-E_{+,X}(z)E_{-+,X}(z)^{-1}E_{-,X}(z),\\
 \Gamma_X(z)&=-E_{-+,X}(z)^{-1},\qquad
 \Gamma_X(z)E_{-+,X}(z)=-1.
 \end{aligned}
\end{equation}
The second line follows by taking the root matrix element: the root
entry of $E_X(z)$ is zero, and the root entries of $E_{+,X}(z)$ and
$E_{-,X}(z)$ are both one. In the one-vertex case this gives
$E_{-+,X}(z)=z$ and $\Gamma_X(z)=-z^{-1}$.

Apply \eqref{eq:grushin-effective}--\eqref{eq:grushin-recovery} with
$X=T$, and label the clusters reached by the $n$ open child edges as
$T_1,\ldots,T_n$. We obtain
\begin{equation}\label{eq:recursion}
 \Gamma_T(z)=-\left(z+\sum_{i=1}^n\Gamma_{T_i}(z)\right)^{-1},\qquad
 \Gamma_T(z)\left(z+\sum_{i=1}^n\Gamma_{T_i}(z)\right)=-1.
\end{equation}
This is the recursion in \cite[Section~2.1, equation~(6)]{Bordenave},
now obtained from the root Grushin problem. We use the product form when taking limits at real energies.
 Earlier uses of this recursion in localization and quantum percolation
 appear in
\cite{AbouChacraThoulessAnderson,Harris}.
Let $\mu_T$ be the spectral measure of $A_T$ at $\delta_o$.
Then
\begin{equation*}
 \Gamma_T(z)=\int_{I_k}\frac{\dd\mu_T(\lambda)}{\lambda-z},\qquad
 \Imn\Gamma_T(z)=\Imn z\int_{I_k}\frac{\dd\mu_T(\lambda)}{|\lambda-z|^2}>0.
\end{equation*}

\medskip\noindent\textbf{Boundary values.}
Let $\mathcal F_k$ contain one representative of each nonempty finite
rooted tree of maximum degree at most $k+1$, up to relabelling that
preserves the root, and set
\[
 \Sigma_{\rm fin}=\bigcup_{F\in\mathcal F_k}\Spec(A_F).
\]
For an oriented edge $y\to x$, let $T_{y\to x}$ be the open cluster
rooted at $y$ after deleting $\{x,y\}$. Write
\[
 G_x(z)=\langle\delta_x,(A_\omega-z)^{-1}\delta_x\rangle,\qquad
 \Gamma_{y\to x}(z)=\Gamma_{T_{y\to x}}(z),\qquad z\in\Hh.
\]
The corresponding adjacency operators form a countable family with norms
at most $k+1$, and $\Sigma_{\rm fin}$ is countable. Section~\ref{sec:boundaryfacts}
therefore gives a deterministic Borel set $\mathcal E_p\subset I_k$ with
\begin{equation}\label{eq:goodenergies}
 |I_k\setminus\mathcal E_p|=0,\qquad
 \mathcal E_p\cap\Sigma_{\rm fin}=\varnothing,
\end{equation}
such that, for each $E\in\mathcal E_p$, all $G_x(E+i0)$ and
$\Gamma_{y\to x}(E+i0)$ are finite on one probability-one set.
The one-vertex tree gives $0\in\Sigma_{\rm fin}$, hence $0\notin\mathcal E_p$.

Fix $E\in\mathcal E_p$; all limits below are taken on this common
probability-one set, and undefined boundary values are set to zero.
For the finite or empty branch $F$ of Lemma~\ref{lem:conditioning},
with root $o_F$ when nonempty, define
\begin{equation}\label{eq:decorationenergy}
 m_F(E)=
 \begin{cases}
 0,&F=\varnothing,\\
 \langle\delta_{o_F},(A_F-E)^{-1}\delta_{o_F}\rangle,&F\ne\varnothing,
 \end{cases}
 \qquad Y(E)=m_F(E)\in\R.
\end{equation}
These values are finite and real because $E\notin\Sigma_{\rm fin}$.
The role of Bordenave's Section~2.2 at this point is structural, not a
boundary-value argument. In his notation, equation~(9) separates the
resolvents of the infinite skeleton branches from the finite contribution
$V(z)$, and equation~(10) writes $V(z)$ as the sum of the resolvents of the
extinct child branches \cite[Section~2.2, equations~(9)--(10)]{Bordenave}.
The existence of the boundary values used here instead comes from
Section~\ref{sec:boundaryfacts}. Our Grushin identities allow us to pass to
those scalar boundary values without inverting any branch operator at the
real energy. Taking scalar limits in
\eqref{eq:grushin-effective} and \eqref{eq:grushin-recovery} gives
\[
 E_{-+,T}(E+i0)=E+\sum_{i=1}^n\Gamma_{T_i}(E+i0),\qquad
 \Gamma_T(E+i0)E_{-+,T}(E+i0)=-1.
\]
Thus $E_{-+,T}(E+i0)\ne0$. 

\medskip\noindent\textbf{Conditioning on survival.}
Work under $\PP_p(\cdot\mid |T|=\infty)$ and put
$W=\Gamma_T(E+i0)$. For each $1\leq j\leq k$, take independent copies
$W_1,\ldots,W_j$ of $W$ and $Y_1(E),\ldots,Y_{k-j}(E)$ of $Y(E)$,
with the two families being independent. Lemma~\ref{lem:conditioning} and
\eqref{eq:recursion} give the conditional recursion below. It is the
fixed-$k$ binomial version of Bordenave's equations~(9)--(10): conditional
on $K=j$, his $N'_s$ equals $j$, the variables $G_x^s$ become the copies
$W_1,\ldots,W_j$, and his finite contribution $V$ becomes $D_j$. Our
$D_j$ also includes closed child positions, each of which contributes zero
through \eqref{eq:decorationenergy}:
\begin{equation}\label{eq:boundaryrec}
 \begin{gathered}
 \mathcal L(W\mid K=j)
 =\mathcal L\bigl(g_{E+D_j}(W_1+\cdots+W_j)\bigr),\\
 D_j=\sum_{r=1}^{k-j}Y_r(E),\qquad D_k=0,\qquad
 g_t(w)=-\frac1{t+w}.
 \end{gathered}
\end{equation}
Closed edges contribute zero by \eqref{eq:decorationenergy}; the preceding
product identity ensures that the denominator is nonzero almost surely.

For a probability measure $\mu$ on $\Hh$, define $P$ and $Q$ by
\begin{equation}\label{eq:PQ}
 P\mu=\E(g_{E+D_1})_\#\mu,\qquad
 Q\mu=\sum_{j=2}^k\frac{\pi_j}{q}
              \E(g_{E+D_j})_\#(\mu^{*j}).
\end{equation}
Here the expectations average over the finite branches in $D_j$:
$P$ corresponds to one infinite child, and $Q$ to at least two.

\begin{lemma}\label{lem:zeroone}
For every $E\in\mathcal E_p$, either $W\in\R$ almost surely, or its law $\nu$ is a probability measure on $\Hh$ satisfying
\begin{equation}\label{eq:fixedpoint}
 \nu=(1-q)P\nu+qQ\nu.
\end{equation}
\end{lemma}
\begin{proof}
The spectral representation gives $\Imn W\geq0$. The finite-branch contributions
 are real and the denominator at the boundary is nonzero, so
\eqref{eq:boundaryrec} implies, conditional on $K=j$,
\[
 \Imn W\overset d=
 \frac{\sum_{i=1}^j\Imn W_i}{|E+D_j+\sum_{i=1}^jW_i|^2}.
\]
Write $\theta=\PP(\Imn W=0)$. Independence and nonnegativity give
$\PP(\Imn W=0\mid K=j)=\theta^j$. Therefore
\begin{equation}\label{eq:zeroone}
 \theta=\sum_{j=1}^k\pi_j\theta^j,\qquad
 0=\sum_{j=2}^k\pi_j\theta(1-\theta^{j-1}).
\end{equation}
If $0<\theta<1$, the right-hand side is at least
$q\theta(1-\theta)>0$, contradicting \eqref{eq:zeroone}.
Thus $\theta\in\{0,1\}$. The case $\theta=1$ gives real $W$.
If $\theta=0$, its distribution $\nu$ is supported in $\Hh$, and
\eqref{eq:boundaryrec} and \eqref{eq:PQ} give
\[
 \nu=\sum_{j=1}^k\pi_j\E(g_{E+D_j})_\#(\nu^{*j})
     =(1-q)P\nu+qQ\nu.\qedhere
\]
\end{proof}

\medskip\noindent\textbf{Estimates for the averaging operators.}
Let $\zeta$ be the distribution of $g_{E+D_1}$ on
$\operatorname{PSL}_2(\R)$. On $L^2(\Hh,m)$, use $P$ also for the
average of the operators $T_g$ from \eqref{eq:smoothing}.
Area preservation makes this the same action on densities. For a
probability measure $\mu$ on $\Hh$ and $t>0$, \eqref{eq:commutation} gives
\begin{equation}\label{eq:PonL2}
 \begin{gathered}
 Pf=\int T_gf\dd\zeta(g)=\E T_{g_{E+D_1}}f,\\
 \norm P_{2\to2}\leq1,\qquad S_t(P\mu)=P(S_t\mu).
 \end{gathered}
\end{equation}
Indeed, finite branches have only countably many possibilities. Thus
$\zeta=\sum_n c_n\delta_{g_n}$, where $c_n\geq0$ and $\sum_n c_n=1$,
and the integral is an $L^2$-convergent sum since
\[
 \sum_n\|c_nT_{g_n}f\|_2=\sum_n c_n\|f\|_2=\|f\|_2.
\]
The same commutation identity and the convolution bound
\eqref{eq:jsum} give, for any probability measure $\nu$ on $\Hh$,
\begin{equation}\label{eq:Qbound}
 \begin{aligned}
 S_R(Q\nu)
 &=\sum_{j=2}^k\frac{\pi_j}{q}
       \E\bigl[T_{g_{E+D_j}}S_R(\nu^{*j})\bigr] \text{ and }
 \norm{S_R(Q\nu)}_2
 \leq\sum_{j=2}^k\frac{\pi_j}{q}\norm{S_R(\nu^{*j})}_2
 \leq B_k\norm{S_R\nu}_2.
 \end{aligned}
\end{equation}

\subsection{Two finite decorations and a uniform contraction}
The second main estimate is the strict contraction coming from only two
finite decorations. Their probabilities and Lemma~\ref{lem:relativegap}
give the uniform two-step bound below. The constants are kept explicit
only to produce a positive interval in Theorem~\ref{thm:tree}; no attempt
is made to optimize them.

Define the constants
\begin{equation}\label{eq:contractionconstants}
 N_k=3(k+1),\qquad
 \gamma_k=\frac{k-1}{k}\left(1-\frac{3}{2k}\right)^{5(k-1)},\qquad
 \rho_k=\sqrt{1-\gamma_k\frac{4-2\sqrt3}{N_k^2}}\in(0,1).
\end{equation}
\begin{proposition}\label{prop:contraction}
Let $1/k<p\leq3/(2k)$ and $E\in\mathcal E_p$. With $\rho_k$
defined in \eqref{eq:contractionconstants}, the operator $P$ in
\eqref{eq:PonL2} satisfies, for every $f\in L^2(\Hh,m)$,
\begin{equation}\label{eq:Pcontraction}
 \norm{P^2f}_2\leq\rho_k\norm f_2.
\end{equation}
\end{proposition}
\begin{proof}
Conditional on $K=1$, let $F_1,\ldots,F_{k-1}$ be the finite side
branches. All probabilities in this proof use their conditional product
distribution. Define the events $\mathcal D_0,\mathcal D_1$ below; a dotted union
denotes a union of pairwise disjoint events:
\[
 \begin{aligned}
 \mathcal D_0&=\bigcap_{r=1}^{k-1}\{F_r=\varnothing\},\text{ and }
 \mathcal D_1&=\mathop{\dot\bigcup}_{r=1}^{k-1}
       \left(\{|F_r|=1\}\cap\bigcap_{\ell\ne r}\{F_\ell=\varnothing\}\right).
 \end{aligned}
\]
The product distribution in Lemma~\ref{lem:conditioning} gives
\[
 \begin{aligned}
 a:=\PP(\mathcal D_0)
 &=\prod_{r=1}^{k-1}\frac{1-p}{\eta}
 =\frac{(1-p)^{k-1}}{\eta^{k-1}},\\
 b:=\PP(\mathcal D_1)
 &=\sum_{r=1}^{k-1}\frac{p(1-p)^k}{\eta}
                      \prod_{\ell\ne r}\frac{1-p}{\eta}=(k-1)\frac{p(1-p)^k}{\eta}
       \left(\frac{1-p}{\eta}\right)^{k-2}
 =\frac{(k-1)p(1-p)^{2k-2}}{\eta^{k-1}}.
 \end{aligned}
\]
For the one-vertex graph $\bullet$,
\[
 A_\bullet=[0],\qquad m_\varnothing(E)=0,\qquad m_\bullet(E)=-E^{-1}.
\]
Therefore,
\[
 \begin{aligned}
 \mathcal D_0:&\quad D_1=0,\qquad M_0(w)=-\frac1{E+w}=g_E(w),\\
 \mathcal D_1:&\quad D_1=-E^{-1},\qquad
 M_1(w)=-\frac1{E+w-E^{-1}}=g_E(h_E(w)).
 \end{aligned}
\]
The matrices for these maps are given in \eqref{eq:commonmatrices}.
For two independent choices of finite side branches, the three disjoint labelled
events and their maps are
\[
 \begin{array}{c|c|c}
 \text{event}&\text{probability}&\text{map}\\ \hline
 (\mathcal D_0,\mathcal D_0)&a^2&G_0=M_0M_0=g_E^2\\
 (\mathcal D_0,\mathcal D_1)&ab&G_1=M_0M_1=g_E^2h_E\\
 (\mathcal D_1,\mathcal D_0)&ab&G_2=M_1M_0=g_Eh_Eg_E
 \end{array}
\]
Composing with the inverse of the first map gives
\[
 G_0^{-1}G_1=g_E^{-2}g_E^2h_E=h_E,\qquad
 G_0^{-1}G_2=g_E^{-2}g_Eh_Eg_E=g_E^{-1}h_Eg_E=\ell_E.
\]
The two-step average is
\[
 P^2f=\int\!\int T_gT_hf\,d\zeta(g)\,d\zeta(h)
     =\int\!\int T_{gh}f\,d\zeta(g)\,d\zeta(h).
\]
Let $Z=T_gT_h$ for two independent maps with distribution $\zeta$,
and let $Z'$ be an independent copy. Then $\E Zf=P^2f$, and the
Hilbert-space variance identity gives
\begin{equation}\label{eq:variance-expanded}
 \|f\|_2^2-\|P^2f\|_2^2=\frac12\E\|Zf-Z'f\|_2^2.
\end{equation}
Keep only the pairs of labelled events producing $(G_0,G_1)$ and
$(G_0,G_2)$ and their reverses. Combining each pair with its reverse
gives the coefficient $a^2\cdot ab=a^3b$.
Unitarity of $T_{G_0}$ and Lemma~\ref{lem:relativegap}, with $D=k+1$,
therefore give
\begin{equation}\label{eq:variancegap}
 \begin{aligned}
 \norm f_2^2-\norm{P^2f}_2^2
 &\geq a^3b\bigl(\|T_{h_E}f-f\|_2^2
                    +\|T_{\ell_E}f-f\|_2^2\bigr)\geq a^3b\frac{4-2\sqrt3}{N_k^2}\norm f_2^2.
 \end{aligned}
\end{equation}
Other events may produce the same maps; only the three labelled
events are needed for this lower bound.
Finally, $\eta\leq1$ and $1/k<p\leq3/(2k)$ imply
\begin{equation}\label{eq:gammak}
 a^3b=\frac{(k-1)p(1-p)^{5(k-1)}}{\eta^{4(k-1)}}
 \geq\frac{k-1}{k}\left(1-\frac{3}{2k}\right)^{5(k-1)}=\gamma_k.
\end{equation}
For $k\geq2$,
\[
 0<1-\frac3{2k}<1,\qquad 0<\gamma_k<1,\qquad
 0<\frac{4-2\sqrt3}{N_k^2}<1.
\]
Thus $0<\rho_k<1$. Equations \eqref{eq:variancegap}--\eqref{eq:gammak} give
\[
 \|P^2f\|_2^2
 \leq\left(1-a^3b\frac{4-2\sqrt3}{N_k^2}\right)\|f\|_2^2
 \leq\rho_k^2\|f\|_2^2.
\]
Taking square roots proves \eqref{eq:Pcontraction}.
\end{proof}

\begin{proof}[Proof of Theorem~\ref{thm:tree}]
With $B_k$ from Proposition~\ref{prop:convolution} and $\rho_k$ from \eqref{eq:contractionconstants}, we choose
\begin{equation}\label{eq:epsilon}
 \varepsilon_k=\min\left\{\frac1{2k},\frac{1-\rho_k}{4kB_k}\right\}>0.
\end{equation}
We fix $1/k<p\leq1/k+\varepsilon_k$ and $E\in\mathcal E_p$. Then
\[
 \frac1k<p\leq\frac1k+\varepsilon_k
 \leq\frac1k+\frac1{2k}=\frac3{2k} \text{ and also }
 0<kp-1\leq k\varepsilon_k
 \leq\frac{1-\rho_k}{4B_k}.
\]
Proposition~\ref{prop:contraction} applies. If the second case of
Lemma~\ref{lem:zeroone} held, its probability law $\nu$ on $\Hh$ would
satisfy \eqref{eq:fixedpoint}. The operators in \eqref{eq:PQ} satisfy
the hypotheses of Proposition~\ref{prop:fixedpoint-exclusion} with
$\rho=\rho_k$ and $B=B_k$, by \eqref{eq:PonL2}, \eqref{eq:Qbound} and
\eqref{eq:Pcontraction}. Moreover, \eqref{eq:qbound2} and
\eqref{eq:epsilon} give
\[
 q\leq kp-1\leq k\varepsilon_k\leq\frac{1-\rho_k}{4B_k},\qquad
 \rho_k+2qB_k\leq\rho_k+\frac{1-\rho_k}{2}
 =\frac{1+\rho_k}{2}<1.
\]
The criterion therefore excludes such a law $\nu$.
The resolvent cannot have positive imaginary part almost surely. Lemma~\ref{lem:zeroone}
therefore gives
\[
 \PP_p\bigl(\Gamma_T(E+i0)\in\R\mid |T|=\infty\bigr)=1
 \qquad(E\in\mathcal E_p).
\]
For every oriented edge $y\to x$, the rooted forward trees have the same distribution:
\[
 \mathcal L(T_{y\to x})=\mathcal L(T),\qquad
 \PP_p(|T_{y\to x}|=\infty)=s.
\]
On the finite-cluster event, \eqref{eq:goodenergies} gives
\[
 \Gamma_{y\to x}(E+i0)=m_{T_{y\to x}}(E)\in\R.
\]
On the infinite-cluster event, the conditional resolvent value is real
 almost surely, as just proved. Splitting into these disjoint events gives
\[
 \begin{aligned}
 \PP_p(\Imn\Gamma_{y\to x}(E+i0)\ne0)
 &=(1-s)\PP_p(\Imn\Gamma_{y\to x}(E+i0)\ne0\mid |T_{y\to x}|<\infty)\\
 &\quad+s\PP_p(\Imn\Gamma_{y\to x}(E+i0)\ne0\mid |T_{y\to x}|=\infty)
 =0.
 \end{aligned}
\]
By countability of the oriented edges,
\begin{equation}\label{eq:realcavities}
 \begin{gathered}
 \PP_p(\exists\,y\to x:\Imn\Gamma_{y\to x}(E+i0)\ne0)
 \leq\sum_{y\to x}\PP_p(\Imn\Gamma_{y\to x}(E+i0)\ne0)=0,\\
 \PP_p\!\left(\Gamma_{y\to x}(E+i0)\in\R
                    \text{ for every }y\to x\right)=1.
 \end{gathered}
\end{equation}
For $x\in V$, let $C_x$ be its open cluster, rooted at $x$.
Because $A_\omega$ is the orthogonal sum of the cluster adjacency
operators, $G_x(z)=\Gamma_{C_x}(z)$. The root restriction and extension
maps are $R_{+,C_x}u=u(x)$ and $R_{-,C_x}c=c\delta_x$.
Apply the Grushin problem \eqref{eq:grushin-root} with $X=C_x$.
Removing $x$ from $C_x$ leaves exactly the branches $T_{y\to x}$
for which $y\sim x$ and $\omega_{xy}=1$. Thus
\eqref{eq:grushin-effective} gives
\[
 E_{-+,C_x}(z)
 =z+\sum_{\substack{y\sim x\\\omega_{xy}=1}}\Gamma_{y\to x}(z)
 =z+\sum_{y\sim x}\omega_{xy}\Gamma_{y\to x}(z).
\]
The Grushin recovery formula \eqref{eq:grushin-recovery} gives
\begin{equation}\label{eq:fullresolvent}
 \begin{aligned}
 G_x(z)&=-E_{-+,C_x}(z)^{-1}
       =-\left(z+\sum_{y\sim x}\omega_{xy}\Gamma_{y\to x}(z)\right)^{-1},\\
 G_x(z)E_{-+,C_x}(z)&=-1.
 \end{aligned}
\end{equation}
For fixed $E\in\mathcal E_p$, all relevant scalar boundary values are
finite almost surely. Define
\[
 d_x(E)=E_{-+,C_x}(E+i0)
       =E+\sum_{y\sim x}\omega_{xy}\Gamma_{y\to x}(E+i0).
\]
By \eqref{eq:realcavities}, $d_x(E)\in\R$. Taking scalar limits in the
product identity of \eqref{eq:fullresolvent}, before dividing, gives
\[
 G_x(E+i0)d_x(E)=-1,\qquad
 d_x(E)\ne0,\qquad G_x(E+i0)=-d_x(E)^{-1}\in\R.
\]
Let $\mathcal N(\omega)\subset I_k$ be the set of energies where at least
one $G_x$ fails to have a finite real boundary value. The joint measurability
established in Section~\ref{sec:boundaryfacts} and countability of $V$
allow us to apply Tonelli. For each $E\in\mathcal E_p$, the preceding
argument gives $\PP_p(E\in\mathcal N)=0$; since
$|I_k\setminus\mathcal E_p|=0$, we obtain
\[
 \int_{\Omega_G}|\mathcal N(\omega)|\,d\PP_p(\omega)
 =\int_{I_k}\PP_p(E\in\mathcal N)\,dE=0.
\]
Thus, almost surely, every $G_x$ has finite real boundary values for
almost every energy in $I_k$. Outside $I_k$ the resolvents are real by
\eqref{eq:energyinterval}. The spectral criterion
\eqref{eq:acdensitycriterion}, applied to the coordinate basis
$(\delta_x)_{x\in V}$, gives $\mathcal H_\ac(A_\omega)=\{0\}$.
Each cluster space is a reducing subspace of $A_\omega$, so the same
conclusion holds simultaneously for all open clusters.
\end{proof}

\begin{proof}[Proof of Corollary~\ref{cor:threshold}]
Theorem~\ref{thm:tree} and \eqref{eq:onset} give
$p_\ac(k)\geq1/k+\varepsilon_k>p_c=1/k$.
For the upper bound, Bordenave's Theorem~3 and its application to
bond percolation in \cite[Section~1.2]{Bordenave} give
$p_{\rm B}(k)\in(1/k,1)$ such that $A_{C_o}$ has a nontrivial
absolutely continuous part with positive probability whenever
$p_{\rm B}(k)<p<1$. This event implies $|C_o|=\infty$, since finite
trees have only eigenvalues. It therefore also has positive probability
conditional on $|C_o|=\infty$. Taking $p=(1+p_{\rm B}(k))/2$ in
\eqref{eq:onset} gives $p_\ac(k)\leq(1+p_{\rm B}(k))/2<1$.
\end{proof}

\section*{Statement on the use of artificial intelligence}

The authors used ChatGPT (OpenAI) during the preparation of this manuscript
as a research and writing aid. In particular, ChatGPT pointed the authors to
the two-step $P^2$-contraction mechanism used in
Proposition~\ref{prop:contraction}: passing from the one-step support to
two-step products produces two relative parabolic transformations with
distinct fixed points, allowing the contraction to be quantified through a
free-group spectral gap. ChatGPT suggested this specific mechanism and
assisted in deriving the resulting quantitative proposition from the
classical ideas of Shubin--Vakilian--Wolff and Wolff together with Kesten's
spectral-gap estimate. The proposition in its present form does not appear
verbatim in those references.

More generally, ChatGPT was used throughout the preparation of the
manuscript for editorial assistance, including reorganizing arguments,
improving exposition,
checking that notation was introduced before use, and identifying relevant literature. The authors
checked the mathematical arguments, citations, and final text and takes full
responsibility for the contents of the manuscript.

\end{document}